\documentclass[a4paper]{article}

\usepackage{amsmath}
\usepackage{amsfonts,
            amssymb, amsthm}            
\usepackage{amsrefs}                    
\usepackage{paralist}                   
\usepackage{booktabs}                   
\usepackage{multirow}                   
\usepackage{algorithm}                  
\usepackage{graphicx}

\newtheorem{definition}{Definition}
\newtheorem{theorem}{Theorem}

\everymath{\displaystyle}

\newcommand{\mult}{\mathrm{mult}}
\newcommand{\ie}{\emph{i.e.}}
\newcommand{\eg}{\emph{e.g.}}

\usepackage{xcolor}
\newcounter{ToDoCounter}

\begin{document}

\title{The Computational Universality of Metabolic Computing}
\author{Ricardo~Henrique~Gracini~Guiraldelli~
        and~
        Vincenzo~Manca\\
        University of Verona\\
        \url{{ricardo.guiraldelli, vincenzo.manca}@univr.it}}

\maketitle

\begin{abstract}
    System and synthetic biology are rapidly evolving systems, but both lack tools such as those used in engineering environments to shift the their focus from the design of parts (details) to the design of systems (behaviors); to aggravate, there are insufficient theoretical justifications on the computational limits of biological systems. To diminish these deficiencies, we present theoretical results over the Turing-equivalence of metabolic systems, defines rules for translations of algorithms into metabolic P systems and presents a software tool to assist the task in an automatic way.
\end{abstract}

\section{Introduction}
\label{sec:introduction}

There is an increasing interesting, in the academic community, in the modeling of existing biological cells as well as synthesis of new ones. Systems biology and synthetic biology are parallelly developing themselves with great speed but, in some sense, apart and lacking mutual synergy. Nonetheless, a duality between these fields may be established and both may benefit of the same set of mathematical, computer science and engineering techniques to improve their analysis or design tasks; for instance, \emph{molecular computation}~\cite{Sarpeshkar2010}*{Section~20.7.5}, in which the present work is certainly part of, is one of those tools.

Through the usage of the Church-Turing thesis~\cite{Lewis1997}*{Section~5.1}~\cite{Sipser2012}*{p.~183}~\cite{Cooper2004}*{Section~2.5} and algorithmic construction, it will be shown that a particular kind of Metabolic P (MP) systems~\cite{Manca2013}*{Chapter~3} called \emph{positively controlled MP systems (MPPC)} has the same computational power of a \emph{register machine} and, hence, of a Turing machine, along with a demonstration of the equivalence relation between these models. Then, an outline of the software tool that converts a register machine to MPPC is presented. At last, there is a discussion on the importance of the results of this work as a connector between systems and synthetic biology and introduction of a computational perspective over these fields is stressed.

For this purpose, the present text is divided as follow: Section~\ref{sec:computational-devices} and Section~\ref{sec:mp-systems} introduces the bases of universal computing devices and MP systems; Sections~\ref{sec:mr-to-mppc}~and~\ref{sec:translation-strategy} develop the equivalence of the both systems. Section~\ref{sec:software} describes the developed tool and presents an example of translation of register machine to MPPC grammar. Finally, at the last section, we discuss the advantages of the present formalism to both systems and synthetic biology.
\section{Universal Computing Devices}
\label{sec:computational-devices}

The concept of (universal) \emph{computability} is closely tied to concept of \emph{algorithm}~~\cite{Lewis1997}*{p.~246} that, in turn, is tied with the concept of \emph{Turing machine}~\cite{Lewis1997}*{p.~246}~\cite{Sipser2012}*{Definition~5.17}, which may be considered the main computational device and stands out by its formal and construction simplicity.

Although simple, there are other computational devices equivalent to Turing machine that are more convenient to use at certain modeling occasions. Some examples of equivalent devices are \emph{recursive functions}, \emph{grammars} and \emph{register machine}. In the present work, a particular kind of the latter will be used due to its simple instruction set and similarity to electronic computer architecture.

\subsection{Register Machines}

According to Minsky~\cite{Minsky1967}*{Section~11.1}, a \emph{register machine} (or \emph{program machine}) is composed of finite number of \emph{registers} which has infinity capacity, a small set of \emph{operations} over the registers and a hard coded \emph{program}, which seems an indexed sequence of \emph{instructions}. At last, an instruction is defined as a triple (\emph{operation}, \emph{register reference}, \emph{list of instructions}).

The basic operations of a register machine must reflect those ones that defines recursive function, the (free) movement of the head of the Turing machine across the tape, a signal of end-of-computation and limit its operations to the set \(\mathbb{N}\). Hence, the base register machine model in~\cite{Minsky1967}*{Table~11.1-1} define four operations: \emph{zero}, \emph{successor}, \emph{decrements or jump} and \emph{halt}.

The present work, nevertheless, uses a variation~\cite{Shepherdson1963}*{Section~4, p.~225}~\cite{Manca2012b}*{p.~293} of the register machine defined by Minsky~\cite{Minsky1967}*{Section~14.1} as following seen.

\begin{definition}[Register Machine]
    A register machine \(\mathcal{R}\) is a computational device defined as
    \[ \mathcal{R} = (R, O, P) \]
    where:
    \begin{enumerate}
        \item \(R = \{R_1, R_2, \ldots, R_n\}\) is a finite set of infinite capacity registers, with \(n \in \mathbb{N}\);
        \item \(O = \{\mathtt{INC}, \mathtt{DEC}, \mathtt{JNZ}, \mathtt{HALT}\}\) is the set of operations;
        \item \(P = (I_1, I_2, \ldots, I_n)\) is the program, with \(n \in \mathbb{N}\).
    \end{enumerate}
    The execution of the program \(P\) always start at the first instruction \(I_1\) and procedures sequentially (unless for programmed execution re-route). 
\end{definition}

The instructions \(I\) of the register machine \(\mathcal{R}\) are special notations that conveniently name operations over addressed registers of \(\mathcal{R}\) according to Definition~\ref{def:instructions}.

\begin{definition}[Instructions]
    Let the content of register \(R_i\) be equal to \(x\). Then, it is possible to define the instructions \(I\) of the register machine \(\mathcal{R}\) as following:
    \begin{enumerate}
        \item \(\mathtt{INC(R_i)} = x + 1\);
        \item \(\mathtt{DEC(R_i)} = \begin{cases} x - 1 &, \text{ if } x > 0\\0 &, \text{ otherwise}\end{cases}\);
        \item \(\mathtt{JNZ(R_i, I_j)}\) change the execution flow of \(\mathcal{R}\) setting \(I_j\) as the next instruction to be executed in case \(x > 0\); otherwise, the execution flow keeps sequential;
        \item \(\mathtt{HALT}\) ends the computation of \(\mathcal{R}\).
    \end{enumerate}
    \label{def:instructions}
\end{definition}

A register machine, as seen above, is a simple but powerful computational device easily translate to digital computer architecture. Now it is time to examine that one inspired by cell metabolism, the \emph{MP systems}.
\section{Metabolic P Systems}
\label{sec:mp-systems}
The cell may be seen a small dimension but complex and dynamical system limited by a membrane which separates the external world from the internal cellular machinery. Acting as an interface, this membrane selectively collects molecules from surroundings of the cell and expels others that were produced or refused inside it.

The process of material exchange interfaced by the membrane has a parallel with systems' theory, in which the cell works as a box (isolated system) for transformation of substances, suggesting the existence of a computational process inside it. Gheorge P\u{a}un, aware of this mechanism, proposed a computational model called P system~\cite{Paun2010} (precursor of \emph{membrane computing}) which is based on membrane interaction in cell systems but, at the same time, mathematically formal and consistent, using the concepts of multiset and rewriting systems for the construction of its formal framework.

Although full of new and interesting ideas for biological modeling, P system still presents mechanism too tied to formal languages that forbids its usage to model real-world metabolic and intra-cellular interaction~\cite{Manca2005}*{p.64}.
Attentive to the necessities of bio-modeling, a new membrane computing computational model based on this system which is named \emph{Metabolic P system} or, for short, \emph{MP system}~\cite{Manca2005}~\cite{Manca2013}*{Chapter~3}. 

The primary goal of the MP system is to deterministically model metabolic processes, serving as a powerful (discrete) mathematical tool for expressing and supporting biological studies in the cell magnifying level; also, it meant to be a computational ``intermediate language'' for easy simulation of the formalized models; at last, it should use promptly understandable notation for potential users unfamiliar with the theoretical computer science jargon commonly found in new computational models.

Strongly influenced by chemical reactions, MP system has a reaction-like notation---that can be seen, also, as a \emph{formal grammar} one---supported by recurrence equations and shifts the focus from pointwise string rewriting to a population transformation through the usage of conventional mole concept (as in chemistry). By the other side, its dynamics is supported by linear algebra and relies on matrix operations for solving the recurrence equation system that characterizes the MP system as a (discrete) dynamical one.
In technical jargon, an MP system can be mathematically represented using the support of a kind of formal grammar (named \emph{MP grammar}). As a grammar object, the MP grammar defines all the rules of an MP system, including the process of multisets (through rules and functions), the elements allowed in the multisets and the initial state of the systems~\cite{Manca2013}*{p.~108}.

\begin{definition}[MP grammar]
    An MP grammar \(G\) is a generative grammar for time series defined as
    \[ G = (M, R, I, \Phi) \]
    where:
    \begin{enumerate}
        \item \(M = \left\{ x_{1}, x_{2}, \ldots, x_{n} \right\}\) the finite set of substances (or metabolites), and \(n \in \mathbb{N}\) the quantities of substances.
        \item \(R = \left\{ {\alpha}_{j} \rightarrow {\beta}_{j} \ \vert \ 1 \leq j \leq m \right\}\) the set of rules (or reactions), with \({\alpha}_{j}\) and \({\beta}_{j}\) multisets over \(M\), and \(m \in \mathbb{N}\) the number of reactions.
        \item \(I = \left(x_{1}[0], x_{2}[0], \ldots, x_{n}[0]\right)\) is the vector of initial values of substances or the metabolic state at initial step (step 0).
        \item \(\Phi = \left\{ {\varphi}_{1}, {\varphi}_{2}, \ldots, {\varphi}_{m} \right\}\) is a set of functions (also called regulators), in which every \({\varphi}_{j} : \mathbb{R}^{n} \mapsto \mathbb{M}\), for \(1 \leq j \leq m\), is associated with a rule \(r_{j} \in R\).
    \end{enumerate}
    \label{def:mp-grammar}
\end{definition}

According to definition~\ref{def:mp-grammar}, \(G\) generates (a set of) time series, each of them representing the ``amount of quantity'' of the substances during the time and its time series is calculated for any time \(t\) if and only if \(\displaystyle \frac{t}{\tau} \in \mathbb{N}\) for a given constant \(\tau\).

Notwithstanding, the rules \({\alpha}_{i} \rightarrow {\beta}_{i} \in R\) depends, as equivalently happens in chemical reactions, on the quantities of the ``substances'' in the system, which can be expressed with the support of two different concepts: one that maps the multiplicity expressed in the rule for a substance to the actual number of molecules (of the substance) in the system, and; the quantity of mass the unit of the multiplicity represents (for a particular substance).

Hence, if the aforementioned three restrictions are that in consideration along with an MP grammar \(G\), it is formally defines a \emph{MP system}.

\begin{definition}[MP system]
    A \emph{MP system} \(M\) is a discrete dynamical defined as
    \[ \mathcal{M} = (G, \tau, \nu, \mu) \]
    with
    \begin{enumerate}
        \item \(G\) being an MP grammar following the definition~\ref{def:mp-grammar};
        \item \(\tau \in \mathbb{R}\), the period (amount of time) of a computational step;
        \item \(\nu \in \mathbb{R}\), the number of molecules that represents the (conventional) mole in the system;
        \item \(\mu \in \mathbb{R}^{n}\) is the vector of the mole masses of substances. 
    \end{enumerate}
    \label{def:mp-system}
\end{definition}

As described in definition~\ref{def:mp-system}, the quantities of the substances are dependent of its previous values and a variational function that may depend on other substances and parameters (in the case of the parametric MP system). This additional variance through time is represented as a recurrent equation which the future value of a substance \(X\) is represented as \(x[i+1] \propto x[i]\).

For the computation of these step values, nonetheless, two mathematical accessories were developed. The first, the \emph{stoichiometric matrix}, is based on the arithmetic executed over chemical reactions to calculate the balance of molecules in a chemical system; the other, \emph{equational metabolic algorithm}, synthesizes the whole computational process specified by the an MP system.

\begin{definition}[Stoichiometric matrix]
    Let \(r_{i} = {\alpha}_{i} \rightarrow {\beta}_{i}\), where \({\alpha}_{i}\) (with an equivalent for \(\beta_{i}\)) is represented as \(\displaystyle \sum{{k}^{+}_{i,j} \times {X}_{j}}\ \vert\  {k}_{i,j} \in \mathbb{N} \wedge X_{j} \in M\).

    Let \(\mult^{+}(X_{j}, r_{i}) = {k}^{+}_{i,j}\) be the multiplicity, for the right side (\({\alpha}_{i}\)) of the rule \(r_{i}\), of the substance \(X_{j}\) in the rule. Similarly, there is \(\mult^{-}(X_{j}, r_{i}) = {k}^{-}_{i,j}\) for the left side (\({\beta}_{i}\)) of the rule.

    A stoichiometric matrix \(\mathbb{A}\), of dimension \(\vert M \vert\ \times \vert R \vert\), has each of its elements defined by
    \[ a_{l,m} = \mult^{+}(X_{l}, r_{m}) - \mult^{-}(X_{l}, r_{m}) \]
    with \(, 1 \leq l \leq \vert M \vert\) and \(1 \leq m \leq \vert R \vert\).

    \label{def:stoichiometric-matrix}
\end{definition}

\begin{definition}[Equational metabolic algorithm---EMA]
    Let \(U[i] = \left( {\varphi}_{1}\left( i \right), {\varphi}_{2}\left( i \right), \ldots, {\varphi}_{m}\left( i \right) \right)^{T}\) be the vector of values, in the time step \(i\), of all regulators, and \(\mathbb{A}\) the stoichiometric matrix.

    The \emph{vector of substance variation at step \(i\)}, \({\Delta}[i]\), is computed by the equation
    \[ {\Delta}[i] = {\mathbb{A}} \times U[i] \]
    so-called \emph{Equational Metabolic Algorithm} whom computes the value of any substance in the time future time step \(i+1\) through the recurrent equation
    \[ X[i+1] = X[i] + {\Delta}[i] \]
    \label{def:ema}
\end{definition}

The above definitions, now, complete specifies the discrete dynamical system and the ways to compute its values.

\subsection{Positively Controlled Metabolic P System}
The cell, motto of the MP systems, keep itself working based on metabolic rules. Those, of biochemical equations nature, require enough quantity the metabolites on the left-hand side of their equations to transform the matter and, hence, keep the metabolic circuitry active. Thus, if enough quantity of all metabolites are provided for all equations, the metabolic circuitry works properly; in the total absence, it does not work. However, in the case which a subset of rules do not have enough quantity of some metabolites, these rules are inactivated while the others keep their activities normally. As a result, the metabolic circuitry of the cell keeps working with a broken chain of reactions.

This kind of operation on cell metabolism happens because the (mathematical) operations in cells (and bio-chemistry, in general) happens in the space of natural numbers, which implies that \emph{half (decimal) quantities} or \emph{debt (negative) quantities} cannot be considered in the (left-hand side of the) rules.

In the MP systems world, the above cell-like behavior is defined in a subclass called \emph{Positively Controlled Metabolic P} systems, or \emph{MPPC}: it is a standard MP systems with few controls to ensure all operations are executed in the set of positive numbers.

\begin{definition}[MPPC Grammar]
    A MPPC grammar \(\mathcal{G}_{+}\) is a standard MP grammar \(\mathcal{G}\) respecting the following restrictions, at each computational step \(s_i\):
    \begin{enumerate}
        \item \(\varphi{\vert}_{s_i} = \begin{cases}\kappa &, \text{ if } \kappa \geq 0\\0&, \text{ otherwise}\end{cases}\), for all \(\varphi \in \Phi\);
        \item \({\sum}_{\varphi{\vert}_{s_i} \in {\Phi}_{x}^{-}}{\varphi} \leq x\), with \({\Phi}_{x}^{-}\) the set of all fluxes related to rules in which the metabolite \(x\) is in the left-hand side of the rule; otherwise, \(\varphi = 0, \forall \varphi \in {\Phi}_{x}^{-}\) at the execution step.
    \end{enumerate}
    \label{def:mppc}
\end{definition}

As it will be seen later in the next sections, this definition is useful not only for cell modeling, but also for establishing a comparison with universally computational devices.
\section{Computational Architecture as Metabolic Systems}
\label{sec:mr-to-mppc}

One of the goals of synthetic biology is to produce programmable metabolic systems that could, just as a designed machine, perform well-defined tasks for a certain objective. Some successful approaches of computer-aided design of biological systems already exists~\cite{Beal2011a}, but they are limited by the expressiveness of digital gates.

Theoretically, nonetheless, we are possible to synthesize metabolic systems more significant (in computational power) than those of Boolean circuits. For the present purpose, a programmable  and Turing-potent (register) machine is implemented in a metabolic P system.

\begin{theorem}[Translation of Register Machine to Positively Controlled MP grammar]
    For any register machine \(\mathcal{R}\) exists an equivalent positively controlled MP grammar \(\mathcal{G}_{+}\).
    \label{thm:mr-to-mppc}
\end{theorem}

\begin{proof}[Proof of Theorem~\ref{thm:mr-to-mppc}]
    Given a register machine \(\mathcal{R} = (R, I, P)\) with \(|R| = r\) and \(|P| = p\), a positively controlled MP grammar \(\mathcal{G}_{+} = (M, Ru, I, \Phi)\) is constructed 
    \begin{enumerate}
        \item adding a metabolite \(R_i\) in the set \(M\) for each register \(R_i \in R\);
        \item adding a metabolite \(I_j\) in the set \(M\) for each of the instructions in \(I_j \in P\);
        \item adding a metabolite \(L_j\) in the set \(M\) for each instruction \(I_j \in P\) of the type \(\mathtt{JNZ}\);
        \item adding a \(HALT\) metabolite in the set \(M\);
        \item defining the initial state of the metabolites \(R_j\) equal to the initial values of the registers \(R_j\), the initial values of all the other metabolites to \(0\) and the initial value of \(I_1\) to \(1\);
        \item adding the rules to \(Ru\) and the fluxes to \(\Phi\) according to the following rules:
          \begin{enumerate}
            \item if \(I_j\) is \(\mathtt{INC}\) or \(\mathtt{DEC}\), then \(I_j \rightarrow I_{j+1} : I_j\);
            \item if \(I_j\) is \(\mathtt{INC(R_i)}\), then \(\emptyset \rightarrow R_i : I_j\);
            \item if \(I_j\) is \(\mathtt{DEC(R_i)}\), then \(R_i \rightarrow \emptyset : I_j\);
            \item if \(I_j\) is \texttt{HALT}, then \(I_j \rightarrow HALT : I_j\);
            \item if \(I_j\) is \(\mathtt{JNZ(R_i, I_k)}\), then 
            \begin{enumerate}
                \item \(I_j \rightarrow L_j : I_j\);
                \item \(L_j \rightarrow I_k : L_j - I_{j+1}\);
                \item \(L_j \rightarrow \emptyset : I_{j+1}\); and,
                \item \(\emptyset \rightarrow I_{j+1} : I_j - R_i\).
            \end{enumerate}
        \end{enumerate}
    \end{enumerate}

    From the rules above, it is possible to notice that \(I_j\) and \(L_j\) instructions controls the execution flow of the system and satisfies \[ HALT + {\sum}_{j = 1}^{p}{I_j} = 1 \] \[ 0 \leq {\sum}_{L_j \in M}{L_j} \leq 1 \] ensuring no two instructions are executed at the same time, but its execution starts from instruction \(I_1\) and proceeds sequentially (or jumps to another one in case of a satisfying \texttt{JNZ} instruction).

    All operations are mappings from and to the \(\mathbb{N}\) set once both \(\mathcal{R}\) and \(\mathcal{G}_{+}\), by definition, restrict their operations to this set.

    At last, when a rule \(I_j \rightarrow HALT\) is performed, the system is stuck in a fixed point since there is no rules for ``exiting'' this state. 
\end{proof}
\section{Translation Strategy}
\label{sec:translation-strategy}

Theorem~\ref{thm:mr-to-mppc} defines the relation between register machine and positively controlled metabolic P system, but it does not develop the intuition for its reasons. For a complete comprehension of this equivalence and, consequently, its impact in correlated fields, this section will focus on the constructive analysis of a MPPC from a register machine (just as proposed by Theorem~\ref{thm:mr-to-mppc}) and discuss the challenges of this procedure.

\subsection{Is It Possible To Translate Register Machines Into MP Systems?}

A register machine is a general purpose computational device presenting the same computational power of a Turing machine; it allows the implementation of any algorithm and can deal with languages in the whole spectrum of the Chomsky hierarchy~\cite{Lewis1997}*{p.~272}. By the other side, MP systems have been show to be equivalent to Petri networks~\cite{Castellini2010} and have been used to model complex systems, from (originally) metabolic systems to mathematical series~\cite{Manca2013}*{Chapter~3}; however, MP systems are not proved to be as powerful as a register machine and an equivalence between these systems cannot be set yet.

Nevertheless, MP systems presents a series of characteristics that induces to the idea of an existing equivalence between both systems. It is important to note that these features, by themselves, do not ensure the existence of the any kind of relation between the formalisms, but solely encourages the investigation through a number of theoretical arguments.

At first, \emph{MP systems are (a particular kind of feedback) dynamical systems}; therefore, it receives a set of states, manipulates them and feeds them back to itself in order to keep the computation going on. It is a computational device with proved restricted power~\cite{Castellini2010}, but with potential for increasing it~\cite{Hinrichsen2005}~\cite{Siegelmann1998} particularly because it does not restrict the set of functions it may works on---\ie, the fluxes may be any computable function according to its definition.

Then, both systems works in a discrete representation of time, computing the elements in well-defined steps instead of continuously. Each step executes a finite number of instructions in a finite amount of time and, more, with a finite representation (\ie, there is no guarantee of infinite precision for every object computed). Although these limitations on time are not strong enough to ensure the equivalence (after all, both Turing machines and finite state machines are discrete but not equivalent), it prevents analyzes that goes beyond the Turing equivalence, named super-Turing~\cite{Gheorghe2011}~\cite{Siegelmann1998}.

At last, MP systems have representations with origins on (or inspired by) formal languages (grammar), rewriting systems and graph theory (MP graph~\cite{Manca2013}*{p.~109 and Figure~3.1}). Not enough, it is part of membrane computing, a wide field with other devices as powerful as (or even more than) the Turing machine~\cite{Calude2004}*{p.~179} (with properties, though, that make them impracticable for daily basis use).

\subsection{What Is The Strategy To Implement?}
\label{sec:strategy}

In order to show that MP systems are general purpose computational devices as powerful as register machines (and, hence, Turing machines), it is enough to show that it is possible to implement in MP systems the same algorithms accepted by the register machine. And to satisfy the generality of these implementations, it suffices to translate the register machine instructions into MP equivalents and show they work properly for a well-defined subset of all acceptable algorithms in that machine.

Expressions of behavior in MP systems are described solely through its MP grammar~\cite{Manca2013}*{Definition~3.1}, which is defined as a quadruple consisting of sets of metabolites, rules, initial states and regulators (or fluxes). Therefore, the definition of the four elements of a MP grammar defines univocally the operational actions of a system; if the definition is relaxed and the initial states are not specified, the relaxed MP grammar defines a whole class of behaviors instead of a specific one.

The strategy, then, is to restrict the definition of a \emph{relaxed MP grammar that models the whole class of problems accepted by a register machine}, populating the sets that compose it according to the next strategies.

\subsubsection{Metabolites Set}
\label{sec:metabolites-set}

The metabolites set \(M\) in a MP system is the equivalent to \(M = X \cup Y\), the union of the state space \(X\) and output value space \(Y\) of dynamical systems~\cite{Hinrichsen2005}. It declares all the variables (states) of the system, those allowed to change its values through the computation process.

Analyzing the register machine instructions~\cite{Shepherdson1963}*{Section~4, p.~225}~\cite{Minsky1967}*{Chapter~11} from Section~\ref{sec:computational-devices} (the unique procedures that alters the states of a register machine), it is possible to recognize and isolate all the variables of the machine just looking to the addressed properties subject to manipulation. These can be classified as
\begin{itemize}
    \item \emph{Registers.} Three out of the four instructions (\texttt{INC}, \texttt{DEC} and \texttt{JNZ}) refers to registers, either for manipulation or query the data. Registers are the external variables~\cite{Hinrichsen2005}*{pp.~74} of the register machine, the ones of interest for an user that feeds the machine with a program.
    \item \emph{Instruction Pointer.} Each of the instructions of the machine are executed sequentially, one at each computational step; however, this serial behavior can be modified with the usage of the \texttt{JNZ} instruction, which ``jumps'' the next execution step to another indexed instruction if an addressed register has its contents different from zero. And in order to keep the track of the instructions to be performed, there is a special kind of register (called \emph{program counter}) whose content is the index of the instruction to be performed. Normally, its initial value is one and it is automatically incremented by one at the end of the execution of the instruction; however, \texttt{JNZ} is able to change it to whatever (valid) index it is necessary, hence qualifying program counter as a register machine's variable. 
    \item \emph{Final State.} To indicate the end of a computational process executed by a register machine, there is a special instruction called \texttt{HALT} which signals the machine to stop its execution because the algorithm in process has finished its operation. Hence, \emph{halt} (\emph{final} or \emph{acceptance state}~\cite{Lewis1997}) is a particular that marks the end of the computational process.
\end{itemize}

\subsubsection{Rules Set}

The rules in a MP systems define the relation between the diverse variables of the system, without pay attention to the value update of the variables. Rules, then, characterizes the operation of the system and are presented in \(\alpha \rightarrow \beta\) structure---which can be read as \emph{\(\alpha\) becomes (or transforms itself into) \(\beta\)}.

In the current scope, the definition of the structural behavior of the general MP system that simulates a register machine is done through the observation of both the register machine and the metabolite set: the former is detailed studied through the standpoint of the instruction set and the possible program flows (flowchart) and describes the behavior being reproduced in the new system; the latter provides the description of all the systems' states (variables), the principal constituent of the rules and defines the three subdivisions of behavior of the register machine: \emph{manipulation of the registers}, \emph{manipulation of the flux of the instructions' execution} and \emph{final state of the computational process}.

\paragraph{Manipulation of the Registers} 

The \texttt{INC} and \texttt{DEC} instructions are dedicated to manipulation of the registers, while \texttt{JNZ} uniquely consults their values.

The instructions that alters the values of the registers, consequently, are those to induce changes in (some of) the system's states (the \emph{external} ones~\cite{Hinrichsen2005}*{p.~74}), expressing one of the (main) behavior of the system; therefore, they are modeled in the MP system through particular but very simple rules: they just add to or remove values from addressed registers.

\begin{table}[h!t]
    \centering
    \caption{MP rules for register machine instructions that manipulate registers' values.}
    \begin{tabular}{cc}
        \toprule
        Register Machine Instruction & MP Rule \\
        \midrule
        \(\mathtt{INC(R_i)}\)       & \(\emptyset \rightarrow R_i\) \\
        \(\mathtt{DEC(R_i)}\)       & \(R_i \rightarrow \emptyset\) \\
        \bottomrule
    \end{tabular}
    \label{tab:mp-rules-register}
\end{table}

The first of the rules of Table~\ref{tab:mp-rules-register} can be interpreted as ``from somewhere the register \(R_i\) receives a value'' (and the second one as ``the value of the register \(R_i\) goes to somewhere'') because of the structure of MP rules: inside a delimited environment (\eg, a cell), the \emph{principle of mass conservation} must be preserved and any addition (or removal) of metabolites quantities that breaks this principle must ``come from (or goes to) somewhere''; \emph{somewhere}, in this context, must be understood as \emph{outside the membrane} delimiting the MP system.

\paragraph{Manipulation of the Flux of Instructions}
\label{sec:rules-flux-instruction}

Not all the states of system are registers; those responsible for the sequential execution of the system and the control of its computational workflow, the \emph{instruction pointers}, represents the computational machinery and hide the trickiest translation process between register machine instructions and MP rules. The evaluation of these variables may be divided in two different subclasses: \emph{sequential execution of instructions} and \emph{manipulation of the execution's workflow}.

The first of the subclasses, \emph{sequential execution of instructions}, consists of situations when the token of active instruction is passed from one instruction to another consecutive. Thus, let \(I = \{I_1, I_2, \ldots\}\) be the set of indexed instructions of a register machine representing an algorithm. The computation of the algorithm represented by \(I\) is defined as the sequential execution \(I_1 \vdash I_2 \vdash \ldots\), if and only if \(I_j\) is a instruction of the type \texttt{INC} or \texttt{DEC}.\footnote{Although \texttt{JNZ} instruction also pass the token in a sequential manner, it is not true for every performance of it; therefore, this instruction is excluded of this subclass and is analyzed later.} Then, from the previous statement, it is trivial to define the following conversion rule:

\begin{center}
    \begin{minipage}{0.8\linewidth}
        Let \(I_j, I_{j+1}\) be the \(j^{\text{th}}\) and \((j+1)^{\text{th}}\) indexed instructions, with \(I_j\) belonging to one of the types \texttt{INC} or \texttt{DEC}. Then, there is a MP rule \(I_j \rightarrow I_{j+1}\) that sequentially executes the instruction \(I_j\) in a MP system.
    \end{minipage}
\end{center}

For the above rule, it is important to notice that there is the restriction of the type of instruction (\texttt{INC} or \texttt{DEC}) solely for the \(I_j\) instruction; the \(I_{j+1}\), on the other hand, can be any of the register machine instructions.

The other subclass, manipulation of the execution's workflow, comprises the situations in which the computation of \(I\) does not follow the sequential disposition of the instructions, but changes the execution order under certain circumstances. For this, it is required the \texttt{JNZ} instruction, the exclusive instruction capable of performing such re-arrangements.

Dissecting the \texttt{JNZ} instruction (Definition~\ref{def:instructions}), it is possible to see it has two arguments, a register address and an instruction address, one used to perform an inequality comparison against zero of a register value and the other to redirect the current workflow in case the comparison results are true. \ie if the addressed register \(R_i\) does not have its value equal to zero; since this instruction is inherently composed of two operations, it persuades us to convert it in, at least, two MP rules.

\begin{figure}[h!t]
    \centering
    \caption{Flowchart of the \(\mathtt{JNZ(R_i, I_k)}\) instruction.}
    \includegraphics[scale=0.3]{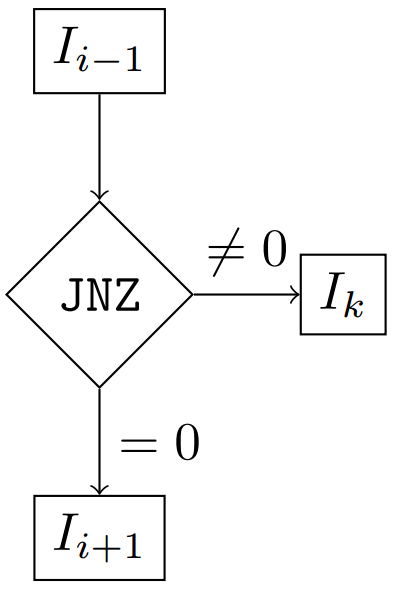}
    \label{fig:flowchart-jnz}
\end{figure}

At first, let us focus in the \emph{comparison} operation of \texttt{JNZ}. Looking at Figure~\ref{fig:flowchart-jnz}, it is easy to see the detour of the sequential instructions' execution happens when the value of the register \(R_i\) differs from zero; this remark suggests most of the trickery will be developed in order to make this execution branch feasible in MP systems.

The artifice is to create a metabolite \(L_j\) that represents the virtual comparison instruction of instruction \(I_j \equiv \mathtt{JNZ(R_i, I_k)}\) and, then, let it support the control of the redirection of the instructions' execution.

Thus, \(L_j\) must receive the token previously in \(I_j\) and, in your own computational step, choose whether \(I_{j+1}\) or \(I_k\) will be the next instruction to be carried out.

For the reception of the token, it follows the same mechanism previously used for sequential execution of the instructions: the \(I_j\) token becomes the \(L_j\) token with the simple rule \(I_j \rightarrow L_j\).

The choice of the instruction path, in turn, makes extensive use of \emph{regulators} (which will be discussed ahead in Section~\ref{sec:regulators-set}) in order to produce its rules. Intuitively, though, it is possible to develop them.

When the contents of \(R_i \neq 0\), \(L_j\) must redirect the execution flow to the instruction \(I_k\) through the transfer of its token to this new instruction; therefore, again, the sequential rule style operates and the produced rule is \(L_j \rightarrow I_k\).

However, when \(R_i = 0\), a problem arises: neither \(I_j \rightarrow I_{j+1}\) nor \(L_j \rightarrow I_{j+1}\) rules may be used to give, in the context of \emph{a metabolite becomes another metabolite}, \(I_{j+1}\) the execution token---the former (\(I_j\)) does not have tokens anymore and the latter (that would be the correct rule) was used to \emph{transform} the token from \(L_j\) to \(I_k\). The solution, then, is to create two rules to reproduce the behavior of \(L_j \rightarrow I_{j+1}\): one takes the token from \(L_j\) and the other gives it to \(I_{j+1}\).

\[
    L_j \rightarrow I_{j+1} \Rightarrow
    \begin{cases}
        L_j \rightarrow \emptyset \\
        \emptyset \rightarrow I_{j+1}
    \end{cases}
\]

At end, an instruction \(I_j \equiv \mathtt{JNZ(R_i, I_k)}\) becomes four MP rules summarized in Table~\ref{tab:mp-rules-jnz}.

\begin{table}[h!t]
    \centering
    \caption{MP rules for the \(I_j \equiv \mathtt{JNZ(R_i, I_k)}\).}
    \begin{tabular}{cc}
        \toprule
        Functional Description & MP Rule \\
        \midrule
        \begin{minipage}{0.5\linewidth}
            \centering
            Pass token to \emph{pointer of virtual comparison instruction} \(L_j\)
        \end{minipage} \vspace{1em}                                   & \(I_j \rightarrow L_j\)            \\

        If \(R_i \neq 0\), pass token to \(I_k\) \vspace{1em}         & \(L_j \rightarrow I_k\)             \\

        \multirow{2}{*}{If \(R_i = 0\), pass token to \(I_{j+1}\)}    & \(L_j \rightarrow \emptyset\)      \\
                                                                      & \(\emptyset \rightarrow I_{j+1}\)  \\
        \bottomrule
    \end{tabular}
    \label{tab:mp-rules-jnz}
\end{table}

\paragraph{Final State of the Computational Process}

At the end of a computational process, some kind of flag is necessary to signalize its termination, otherwise there is no way to guarantee we are seeing a partial result of the computation or the system's final state. Turing machines, in their formal definition~\cite{Lewis1997}*{Definition~4.1.1}, make this explicit signalization through the \emph{halting states} set, and the register machine uses its \texttt{HALT} instruction; these halting objects, as the name suggests, stops the computational device execution.

By the other hand, MP systems are feedback discrete dynamical systems and the concept of stopping the computation does not exist in this context. In other to imitate this behavior, it is possible to add to the systems' variables an additional state (which we will conveniently call it \emph{halt state}) which marks the end of the systems' computational process and, at the same time, is a \emph{fixed point} of the dynamical system. For this purpose, it is enough the addition of a single halt state that is the destination of the token from the \texttt{HALT} instructions and the absence of rules with the halt state in the left-hand side of it. Hence,

\begin{center}
    \begin{minipage}{0.8\linewidth}
        Let \(I_j\) be a \texttt{HALT} instruction. Then, there is a MP rule \(I_j \rightarrow HALT\) that models the end of the computational process with signalization of the final state.
    \end{minipage}
\end{center}

Clarifying the above definition, a register machine does not restricted your programs to have an unique \texttt{HALT} instruction and, thus, it may have several of them, all of them going to a single halt state. Therefore, the ``token'' of the \(j^{\text{th}}\) instruction \emph{becomes} the signal of end of computation and no other instruction is executed, because:
\begin{inparaenum}[(i)]
    \item the value (``token'') that controls the instruction flow is passed from the instructions (\(I_j\)) to the halt state (HALT) and
    \item there is no rules in which the halt state is consumed (\ie, in the left-hand side).
\end{inparaenum}

In fact, those two arguments also justify the halt of the device as a fixed point of the dynamical systems: with no rules being executed, no modification in the states are performed and, that being so, any other execution step of the MP system produces the same, previous state (\(S_{halt} \vdash S_{halt} \vdash \ldots \vdash S_{halt}\)).

\subsubsection{Regulators Set}
\label{sec:regulators-set}

In a MP system, the metabolites set defines the state variables of the system. The rules set, its skeleton behavior. But the update of the values of the state variables depends on mathematical functions associated with the rules, also known as \emph{regulators}.

The regulators computes the rates in which the metabolites in the left-hand side of rules are transformed into the ones in the right-hand side. In both mathematics and biology, regulators are often called \emph{flux}.

Since does not make sense to study regulators detached from rules, the classification of the equivalent regulators of a register machine will follow similar subdivision used for the rules: \emph{update of the registers' values}, \emph{update of the instruction pointers} and \emph{signalization of end of computation}.

\paragraph{Update of the Registers' Values}

The rules presented in the Table~\ref{tab:mp-rules-register} describes the process of updating a register value, but not the functions which actualizes it; these functions (the regulators) can be easily inferred from the specification of the register machine instructions.

Initially, the \texttt{INC} operation: it is a primitive recursive function which adds a single unit to the actual value, or \(\mathtt{INC(R_i)} = R_i + 1\), which can also be stated, in a regulator perspective, as ``\(\mathtt{INC(R_i)}\) changes the value of \(R_i\) with the constant rate of \(1\)''. The latter declaration, in turn, denotes the value of regulator as the constant \(1\); then, it is easy to define the regular for the \texttt{INC} instruction: \(\varphi_{\mathtt{INC}} = 1\).

Conversely, \texttt{DEC} follows a similar reasoning, \(\mathtt{DEC(R_i)} = R_i - 1\)\footnote{Actually, \(\mathtt{DEC(R_i)} = \max\{R_i - 1, 0\}\) since it is an operation \(\mathbb{N}\mapsto\mathbb{N}\). It is guaranteed by the \emph{positive control} of MPPC systems and will be discussed, again, further in this section.}, which can also be stated as ``\(\mathtt{INC(R_i)}\) changes the value of \(R_i\) with the constant rate of \(-1\)''. Nonetheless, it is also necessary to notice the \texttt{DEC} rule ``consumes'' the \(R_i\) value is  (because \(R_i\) is in the left-hand side of \(R_i \rightarrow \emptyset\)) and, hence, its flux is already negative. Thus, \(\varphi_{\mathtt{DEC}} = 1\).

\paragraph{Update of Instruction Pointers and Final State}

As with the case of the update of the registers' values, the behavior of the instruction pointers are well defined in the Section~\ref{sec:rules-flux-instruction}, but no modification to the state representing the instruction pointers are yet defined or performed. But before defining the regulators for their rules, it is important to reflect a little about the nature of the instruction pointers.

The instruction pointers states are flags that, when activated, indicates which of the instructions are being processed in the current computational step; they are the stratagem used to transform the parallel execution of rules in MP system (matrix multiplication at Definition~\ref{def:ema}) into sequential ones as present in the register machines.

The idea of the instruction pointers are similar to the \emph{program counter} of counter and random access machines, with the particularity that each instruction has its own instruction pointer, they are limited to binary values (deactivated and activated, respectively \(0\) and \(1\)) and no two instruction pointers may be activated at the same time, formally defined as\footnote{The presence of \(Halt\) is justified by the necessity of the system be in execution mode, not halted.} \(Halt + {\sum}{I_j} = 1\).

Thus, it is easy to define the regulators of the rules for sequential execution of instructions: for each rule of the form \(I_j \rightarrow I_{j+1}\) or \(I_j \rightarrow Halt\), the regulator is defined as \({\varphi}_{pointer} = I_j\).

(The \(Halt\) state variable may be seen as a particular case of the instruction pointer ones: it marks the end of the computational process instead of the instruction being computed.)

Finally, we must carefully design the regulators for the previously defined \(I_j \equiv \mathtt{JNZ(R_i, I_k)}\) rules. Thus, let us recall the mechanism of this function:

\begin{center}
    \begin{minipage}{0.8\linewidth}
        it compares the value of the register \(R_i\) against the number zero and if it is equal, the next instruction to be executed will be the sequential instruction \(I_{j+1}\), otherwise it will re-route the execution flow to the instruction \(I_k\).
    \end{minipage}
\end{center}

At first, then, \texttt{JNZ} delegates to the virtual comparison instruction \(L_j\) the comparison of the register \(R_i\) value against zero. Modeled by the rule \(I_j \rightarrow L_j\), it is actually a particular kind of sequential instruction execution, but \emph{inside} the \texttt{JNZ} rule. For this reason, its regulator follows the previous defined for the sequential execution rules and assumes the value of \(I_j\).

Next, there are two rules for the case in which \(R_i = 0\): one passes the token to the sequential instruction \(I_{j+1}\) and the other certifies the redirection to the \(I_k\) instruction will not ever occur. The former case, represented by the rule \(\emptyset \rightarrow I_{j+1}\), is a variation of the sequential execution when \(R_i = 0\), what drives us to consider \(I_j\) (sequential execution) and \(R_i\) (register value) as elements for the regulator. Observing the Table~\ref{tab:values-sequential-jnz} of the possible values of these metabolites, it is easy to infer its regulator as \(\varphi_{\emptyset \rightarrow I_{j+1}} = \max(I_j - R_i, 0)\).

\begin{table}[h!t]
    \centering
    \caption{Values of \(I_{j+1}\) depending on the values of \(I_j \equiv \mathtt{JNZ(R_i, I_k)}\) and \(R_i\).}
    \begin{tabular}{ccc}
        \toprule
        \(I_j\) & \(R_i\)       & \(I_{j+1}\)   \\
        \midrule
        0       & 0             & 0             \\
        0       & \(\neq 0\)    & 0             \\
        1       & 0             & 1             \\
        1       & \(\neq 0\)    & 0             \\
        \bottomrule
    \end{tabular}
    \label{tab:values-sequential-jnz}
\end{table}

For the latter of the rules when \(R_i = 0\), \(L_j \rightarrow \emptyset\), its regulator just have to control that the virtual comparison instruction \(L_j\) will not be active (\ie, its value will be zero) when \(I_{j+1}\) is; given that \(L_j\) is solely activated by the \(j^{\text{th}}\) instruction by a sequential-like rule, it is certainty it will not get any other increase of value outside the \texttt{JNZ} workflow and, hence, it has a maximum value of one. Then, it is trivial to infer that \(\varphi_{L_j \rightarrow \emptyset} = I_{j+1}\), establishing \(I_{j+1}\) as the repressor metabolite of \(L_j\).

Finally, the lasting rule of \texttt{JNZ} concerns to the re-route of the execution flow to the \(k^\text{th}\) instruction of the algorithm, \(I_k\). As previously state, a redirection of the flux of execution happens \emph{when the value of the register \(R_i\) is different of zero}, \ie, it is a ``secondary level'' task inside the \texttt{JNZ} workflow, which depends on the rules when \(R_i = 0\). By the other hand, the redirection happens, as state by its rule \(L_j \rightarrow I_k\), as a consequence of the activation of the virtual comparison instruction. Minimizing the metabolites set that influences the re-route sub-task of \texttt{JNZ} and combining all their possible values in the Table~\ref{tab:values-redirection-jnz}, the pattern that arises between this table and Table~\ref{tab:values-sequential-jnz} correctly induces to conclude the regulator for this last rule is \(\varphi_{L_j \rightarrow I_k} = \max\{L_j - I_{j+1}, 0\}\).

\begin{table}[h!t]
    \centering
    \caption{Values of \(I_k\) depending on the values of \(L_j\) and \(I_{j+1}\).}
    \begin{tabular}{ccc}
        \toprule
        \(L_j\) & \(I_{j+1}\)   & \(I_k\)   \\
        \midrule
        0       & 0             & 0         \\
        0       & 1             & 0         \\
        1       & 0             & 1         \\
        1       & 1             & 0         \\
        \bottomrule
    \end{tabular}
    \label{tab:values-redirection-jnz}
\end{table}

It is interesting to notice, however, that both \(\emptyset \rightarrow I_{j+1}\) and \(L_j \rightarrow I_k\) have regulators of the form \(\max\{X - Y, 0\}\), or subtraction in the natural set of the numbers, in order to restrict the regulators to become negative numbers. This guard, also present in the \texttt{DEC} instruction, provides two curious features to the system: 
\begin{enumerate}
    \item the comparison operator (\emph{greater-than}) required in \texttt{JNZ}, ensuring a positive, not null regulator if and only if \(X > Y\); and,
    \item preservation of the coherence of the rules.
\end{enumerate}
While the former characteristic is evident and does not require further explanation, the latter preserves the structural rules of the system: a negative regulator inverses the behavior of the associated rule, \ie \[\left( X \rightarrow Y : -\epsilon\right) \equiv \left( Y \rightarrow X : \epsilon\right)\] Hence, the concession of negative fluxes in the system would allow the modification of its set of the rules, with the modification of a certain rule \(r_i = \{X \rightarrow Y, Y \rightarrow X\}\) depending of its regulator---or, in other words, the system could represent two different programs depending of the positiveness of the regular of \(r_i\).

Furthermore, this restriction pattern goes in accordance with two other (and important) principles: the \emph{positive control}~\ref{def:mppc} of MP grammars, which asserts that fluxes must be greater or equal to zero, and \emph{computability over the natural numbers}~\cite{Lewis1997}*{Definition~4.4.1, Figure~4-19 and Section~4.7}~\cite{Shepherdson1963}*{Section~2}.

\subsubsection{Program and the Initial States}

In Section~\ref{sec:strategy} we introduced the concept of \emph{relaxed MP grammar} as a variation of the MP grammar definition without the requirements to set the initial states. This flexibility allowed to construct the equivalence instruction of register machine in MP as general as it is possible, without restraining the systems to a single, particular behavior: a relaxed grammar defines a whole class of dynamics based on general rules, not on specific simulations dependent of initial states values.

Although it is enough to define an equivalence with register (and Turing) machines through bijection of its composing elements, relaxed MP grammar does not allow the definition of the concept of \emph{program}~\cite{Lewis1997}*{pp.~210--211 and Definition~4.4.1}~\cite{Minsky1967}*{p.~202} because nothing such as \emph{initial configuration}~\cite{Lewis1997}*{Definition~4.4.2} is specified.

The initial configuration solely defines values for the states of the system and no modification of the designed model is undergone. It acts as a constrainer of the device model to ensure the correct execution of the dynamics.

Reviewing the states (metabolites) set of the MP system equivalent to the register machine, Section~\ref{sec:metabolites-set} and later addition in Section~\ref{sec:rules-flux-instruction}, it is possible to partition them under two categories: 
\begin{inparaenum}[(i)]
    \item the \emph{memory units}, which comprise the registers-equivalent metabolites (\(R_i\)); and, 
    \item the \emph{instruction pointers}, containing the instruction pointers \(I_j\), the halting state \(Halt\) and the virtual comparison instruction \(L_j\).
\end{inparaenum}
The first kind is, as in similar computational models~\cite{Hinrichsen2005}*{Section~2.1.1}~\cite{Minsky1967}*{Section~11.1}~\cite{Lewis1997}*{Section~4.4}, memory units that keeps the state of the internal computations but also serves as an interface for input (and, at the end of the computation process, output) data in the program; thus, it must always be initialized with zero values when representing internal computation states, while is open to be freely set to (positive) values when representing storage for input data.

The second kind, instruction pointers, are internal state that control the execution of the MP system as a sequential, computational device. The instruction pointers and virtual comparison instructions, together, represent the program counter~\cite{Lewis1997}*{Defintion~4.4.1} in the MP formalism while the halt state indicates the end of the computation.

In order to correctly execute a program from a register machine, nonetheless, it is also necessary to specify the states of the instruction pointers satisfying the two following constraints:
\begin{align}
    I_1 &= 1 \label{eq:first-instruction}\\
    \sum_{i = 1}^{\left|\mathcal{K}\right|}{K_i} &= 1, \text{ for any } t \in \mathbb{N} \text { and } K_i \in \mathcal{K} = I_j \cup L_j \cup Halt \label{eq:max-one-instruction}
\end{align}

Equation~\eqref{eq:first-instruction} defines a sequence \(S = (I_1, I_2, \ldots)\) for the instructions codified as MP rules merely with the definition of its first element, while the other sequential elements follows from the sequential execution rules (Section~\ref{sec:rules-flux-instruction}). Equation~\eqref{eq:max-one-instruction}, in turn, restricts the execution to solely one instruction (pointer) at each step, guarantees the device is either on work or halted and, finally, forces all the states in the instruction pointer category to have initial values equal to zero (except for \(I_1 = 1\)).

\section{Software}
\label{sec:software}

Following the theoretical development that derived the Theorem~\ref{thm:mr-to-mppc}, a piece of software were codified in order to provide experimental evidence of the validity of the theory, as well as a tool for automatic translation of register machines specifications into MP systems---which, later, expanded for also simulated both systems and ensure their equivalence. Here, though, the bisimulation is not verified, but instead a verification of the final state of both systems is performed.\footnote{The bisimulation is not necessary to be performed in the software due to the result of the Theorem~\ref{thm:mr-to-mppc} which states an equivalence. There is space, though, for this verification if the intermediate steps of both simulations are compared according to~\cite{Sangiorgi2012}*{Definition~1.4.2}.} However, if one of the systems does not halt, it is not possible to perform the comparison of the systems, a expect situation and in accordance to the \emph{halting problem}~\cite{Lewis1997}*{Section~5.3}~\cite{Sipser2012}*{Section~5.1}.

The existing software, a command-line tool available for the major operating systems players, can be divided in three main parts: \emph{translator}, \emph{executor} and \emph{comparator}, as represented in Figure~\ref{fig:software-dataflow}. This modular structure allows an unique flexible feature to the tool to choose the atomic operations to perform for a determined task; also, it provides the opportunity to implement parallel execution of modules and easy addition of new components.

\begin{figure}[h!t]
    \centering
    \caption{Dataflow of the software.}
    \includegraphics[width=0.9\linewidth]{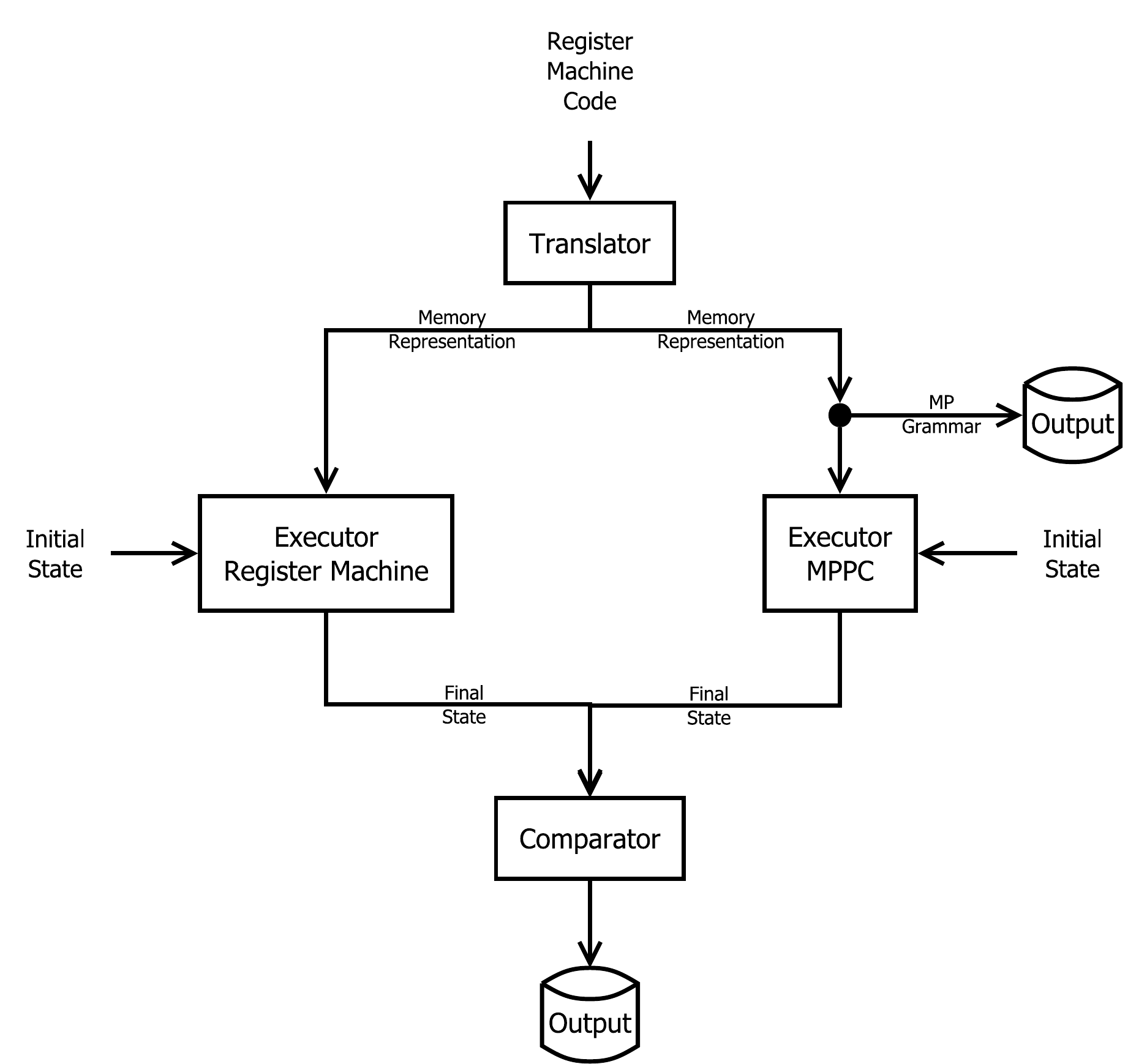}
    \label{fig:software-dataflow}
\end{figure}



When called by the command-line, the software may solely translate the register machine specification into a valid MP grammar (Figure~\ref{fig:software-screenshot-translator}) or generate the equivalent MPPC system to the register machine, simulate both of them and print if they are equivalent based in two criteria: both system halted and the final state of the registers and their respective metabolites in the MP grammar have the same values (Figure~\ref{fig:software-screenshot-simulation}).

\begin{figure}[h!t]
    \caption{Register machine specification for the \(\max(R_1, R_2)\) function.}
    \centering
    \begin{minipage}{15ex}
        \tt\small
        JNZ(1, 4)\\
        INC(4)\\
        HALT\\
        JNZ(2, 7)\\
        INC(3)\\
        HALT\\
        INC(5)\\
        DEC(1)\\
        DEC(2)\\
        JNZ(5, 1)
    \end{minipage}
\end{figure}

\begin{figure}[h!t]
    \centering
    \caption{Translation from register machine to MP grammar with the command-line tool.}
    \includegraphics[width=0.9\linewidth]{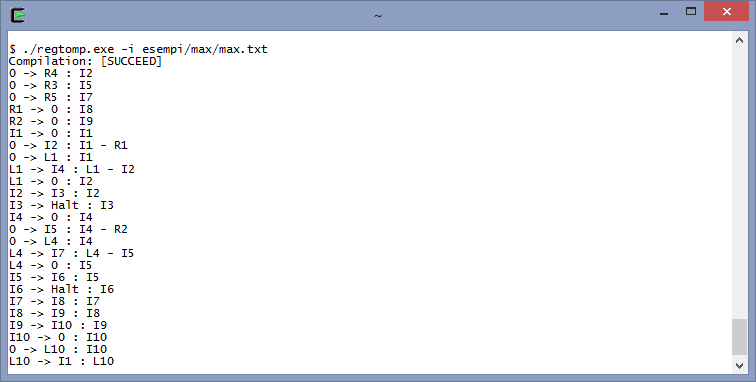}
    \label{fig:software-screenshot-translator}
\end{figure}

\begin{figure}[h!t]
    \centering
    \caption{Simulation and comparison of register machine and equivalent MP grammar with the command-line tool.}
    \includegraphics[width=0.9\linewidth]{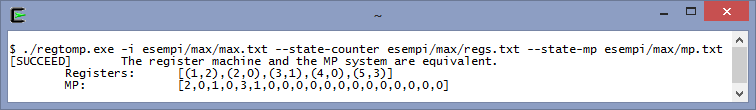}
    \label{fig:software-screenshot-simulation}
\end{figure}

In the previous example, both devices perform the computation \((R_1, R_2, R_3, R_4, R_5) = (5, 3, 0, 0, 0) \overset{*}{\vdash} (2, 0, 1, 0, 3)\) equally, writing \(1\) at \(R_3\) to indicate the value of the register \(R_1 \geq R_2\) (otherwise, \(R_4 = 1\) to indicate \(R_1 < R_2\)).
\section{Conclusion}
\label{sec:conclusion}

Among the several models of cells and metabolism, it is difficult to find one that naturally models the cellular system and provides significant insight to challenge the biological machinery against existing and well-known synthesis theory, such as digital circuits or computer programs. Thus, these models are, usually, very useful for the systems biology analysis, but they lack contribution for the counterpart synthetic biology.

The present work extends the existing Metabolic P system, extensively used for analyzing biological behavior, to computation theory and provides theoretical and software support to convert programmable instructions into positively controlled Metabolic P grammar. Particularly, it stands out from competing models because 
\begin{inparaenum}[(i)]
    \item MP grammar is the formalism for every metabolic process;
    \item it is a proven Turing-powerful model (under positively controlled restriction);
    \item MPPC mimics the cell behavior in the borderline situations;
    \item it provides software tools to transform software specification into MPPC description.
\end{inparaenum}

These features, when combined with hardware synthesis techniques~\cite{Pedroni2004} or database of metabolites (similar to~\cite{Beal2011a}), easily brings the ambition of lab-on-chip~\cite{Gravitz2009} and synthesis biology to the laboratories day-by-day reality.

\bibliographystyle{plain}
\bibliography{computational-universality-metabolic-computing}             

\end{document}